\documentclass[final,onefignum,onetabnum]{siamsials250905}
\newsiamremark{remark}{Remark}

\usepackage{enumitem}

\usepackage[utf8]{inputenc}
\usepackage{amscd,amsfonts,amsopn,amssymb,mathtools}
\usepackage{enumitem}
\usepackage{tikz}
\usepackage{textcomp}
\usepackage{float}
\usepackage{subcaption}
\usepackage{multicol}
\usepackage{mathrsfs}
\usepackage{quiver}
\usepackage[numbers,sort]{natbib}

\usepackage{bbm}

\usepackage{import}

\usepackage{subfiles}

\newcommand{\Z}[0]{\mathbbm{Z}}
\newcommand{\R}[0]{\mathbbm{R}}

\title{A Persistent Homology Pipeline for the Analysis of Neural Spike Train Data}
\author{
Cagatay Ayhan\footnote{Florida State University, Department of Mathematics} \and
Audrey N. Nash\footnotemark[1] \and
Roberto Vincis\footnote{Florida State University, Department of Biological Science, Programs in Neuroscience and Molecular Biophysics} \and
Martin Bauer\footnotemark[1] \and
Richard Bertram\footnote{Florida State University, Department of Mathematics and Programs in Neuroscience and Molecular Biophysics} \and
Tom Needham\footnotemark[1]
}

\date{}

\begin{document}
\maketitle

\begin{abstract}
In this article, we introduce a Topological Data Analysis (TDA) pipeline for neural spike train data. Our framework treats a collection of spike trains recorded from a population of neurons as a metric space, endowed with the Victor-Purpura (VP) distance, to which techniques of  persistent homology are applied. This TDA framework proves capable of identifying stimulus-discriminative structure in simultaneous spike train recordings, even in cases where the discriminating ability of individual neurons is low.\\
\noindent {\bf Relevance to Life Sciences:}
Understanding how the brain transforms sensory information into perception and behavior requires analyzing coordinated neural population activity. Modern large-scale electrophysiology enables simultaneous recording of high-dimensional spike train ensembles, but extracting meaningful information from these datasets remains a central challenge in neuroscience. A fundamental question is how ensembles of neurons discriminate between different stimuli or behavioral states, particularly when many individual neurons exhibit weak or no stimulus selectivity, yet their coordinated activity may still contribute critically to network-level encoding. Here we describe a topological data analysis framework that identifies stimulus-discriminative structure in spike train ensembles recorded from the mouse insular cortex during presentation of deionized water stimuli at distinct non-nociceptive temperatures. We show that population-level topological signatures effectively differentiate oral thermal stimuli even when individual neurons provide little or no discrimination. These findings demonstrate that ensemble organization can carry perceptually relevant information that standard single-unit analysis may miss, contributing to a broader understanding of how network dynamics, rather than isolated feature tuning, support sensory discrimination. Because the framework operates on general point process data, it extends beyond neuroscience to other domains including gene expression timing and financial market dynamics.\\
\noindent {\bf Mathematical Content:}
The framework builds on a mathematical representation of spike train ensembles that enables persistent homology to be applied to collections of unregistered point processes. At its core is the VP distance, which quantifies spike train similarity through costs of spike insertion, deletion, and temporal shifting. Using this metric, we construct persistence-based descriptors that capture multiscale topological features of ensemble geometry. Two key theoretical results support the method: a stability theorem establishing robustness of persistent homology to perturbations in the VP metric parameter, and a probabilistic stability theorem ensuring robustness of topological signatures. This combination of mathematical rigor and methodological flexibility provides a solid foundation for population-level analysis through the lens of topology.
\end{abstract}

\section{Introduction}
The brain's ability to transform sensory inputs into perceptions and guide behavior relies on neural population activity. While individual neurons encode specific features, it is the coordinated activity of neuronal ensembles that encodes perception, cognition, and behavior~\cite{hebb1949organization, buzsaki2004large, yuste2024neuronal}. Modern recording technologies now enable simultaneous measurement of large neuronal ensembles in behaving animals~\cite{Trautmann2025NeuropixelsNHP, Chung2022HumanNeuropixels, Churchland2012Dynamics, Stringer2019HighDimensional}, providing unprecedented access to population-level neural dynamics. Yet extracting meaningful information from these high-dimensional datasets, which often contain heterogeneous and sparse neural responses, remains a central challenge in neuroscience~\cite{Panzeri2015, Yuste2015, Averbeck2006,barth2012experimental, okun2015diverse, hromadka2008sparse}.

A fundamental goal of \emph{in-vivo} neural ensemble analysis is to determine how ensembles of simultaneously recorded neurons discriminate between different stimuli or behavioral states. This presents a significant challenge: many neurons in an ensemble may not exhibit strong individual selectivity to particular stimuli, yet their coordinated activity may still contribute critically to population-level encoding. Here we describe a persistent homology  framework that identifies stimulus-discriminative structure in spike train ensembles, even when responses are distributed across heterogeneous neuronal populations with varying levels of individual selectivity. While we demonstrate this approach using recordings from the mouse insular cortex, the method is broadly applicable to other brain regions and extends naturally to other biological and non-biological point process data such as gene expression timing~\cite{Fromion13,Palande2023, Jethava2011}, or the timing of trades in financial markets~\cite{Engle03}.

Topological Data Analysis (TDA) comprises a collection of techniques, rooted in metric geometry and computational topology, for quantifying the  structure of complex, high-dimensional data. A central tool in TDA is persistent homology, which tracks the birth and death of connected components, loops, and higher dimensional features across scales of a parameter~\cite{edelsbrunner2002topological, ZomorodianCarlsson2005, Carlsson2009}, yielding multiscale descriptors of the dataset that are robust to noise. In neuroscience, TDA has been applied to population activity to reveal nontrivial topological signatures in sensory cortex data~\cite{Singh2008TopologicalAnalysis, giusti2015clique, Guidolin2022-tn}, to classify network dynamical regimes in simulated and recorded spiking systems~\cite{Bardin2018TopologicalExploration},  and more broadly to characterize biological structure across scales~\cite{Amzquita2020,curto2025topological}. In this work, we adapt TDA methodology to ensembles of simultaneously recorded unregistered point process, such as neural spike trains, by introducing a new variation of the mathematical description of these objects. This formulation allows us to obtain several theoretical results that have not been addressed previously, and it provides the foundation for applying a TDA framework to datasets of simultaneously recorded neurons. We demonstrate the utility of the TDA pipeline by using simultaneous recordings from the mouse insular gustatory cortex (GC), the primary sensory cortex for taste~\cite{spector2005representation, vincis2019central} that also represent fluid thermal information from the oral cavity~\cite{Bouaichi2023-ol, Nash2025-ot}. In particular, we show that this approach effectively distinguishes between different neural response patterns to deionized water stimulus presented at different non-nociceptive temperatures. 

%\subsection{Mathematical Content}

Simultaneously recorded spike trains contain rich geometric and topological structure, and applying tools from TDA reveals organizational patterns that standard approaches are not equipped to detect. Central to our framework is the Victor-Purpura (VP) distance~\cite{VictorPurpura1996}, a widely used and biologically meaningful metric for comparing spike trains. Building on this foundation, we establish new stability results that ensure robustness of persistence-based representations to perturbations in both the VP metric parameters (Theorem~\ref{thm: VP stability Theorem}) and in a novel probabilistic neuronal model designed to capture spike train variability (Theorem~\ref{thm:stability_prob_measure}). These results provide rigorous guarantees that support the reliability and interpretability of our method. Although our focus here is on neuronal population activity, the framework is broadly applicable to any data modality consisting of simultaneously recorded point processes. Moreover, except for the VP-specific stability theorem, our methodological contributions place no assumptions on the underlying metric, giving the approach substantial flexibility for use in diverse scientific domains. Together, these results position our framework as a general and theoretically grounded tool for analyzing population-level activity through the lens of topological data analysis 

%\subsection{Structure of the Paper}
The paper is organized as follows. Section \ref{sec:framework} describes our overall framework. It begins  with a mathematical characterization of the spike train data that is the target of the analysis pipeline. We then describe the metric used to characterize the spike train ensembles, the Victor-Purpura distance, as well as the persistent homology pipeline for analyzing ensembles of spike trains. Theoretical results on this framework are provided in Section \ref{sec:results}, focusing on Lipschitz stability results, both with respect to changes in the Victor-Purpura distance parameter and changes in the probability distribution used to mathematically model neural responses. Finally, experimental results provided in Section \ref{sec:experimental results and applications to real data}  illustrate the utility of this pipeline, using both synthetic and biological data. Our synthetic experiments highlight an important strength of TDA-based analysis: even when individual neurons carry no discriminative information, population-wide structure can provide strong signals for classification of responses to stimuli. Next, in biological datasets recorded in the GC from awake, behaving mice, we demonstrate that ensemble-level classification using our topological pipeline often outperforms single-neuron analyses.  We conclude in Section \ref{sec:discussion} with a discussion of the broader implications of this approach, both for the mathematical methodology and for its potential to advance biological understanding of neural population activity.

%%%%%%%%%%%%%%%%%%%%%%%%%%%%%METHODS%%%%%%%%%%%%%%%%%%%%%%%%%%%%%%%
\section{A Persistent Homology Approach for Population-Level Spike Train Analysis}\label{sec:framework}
In this section we provide precise mathematical descriptions of neural spike train populations, which could apply more generally to point processes arising from other applications. We also describe the persistent homology pipeline that we use for analysis of the data. The concepts that we describe here are fundamental in neuroscience, but we give self-contained mathematical formalism for making these concepts precise.  

\subsection{A Mathematical Representation of Spike Train Data}
A neural spike train is a temporal representation of neural activity, which can be represented as a finite set of time-stamps indicating when the neuron emits an action potential (also called a \emph{spike}). Neurons exhibit distinct firing patterns, and their spike trains may vary in response to different stimuli. Motivated by this biological phenomenon, we formally define the concepts of \emph{spike trains} and \emph{neurons} as distinct mathematical entities.

\begin{definition}
    Let $X =\{0,1,\dots,T\}$, with $T \in \mathbb{Z}_{>0}$ be a \emph{discrete time domain}. A \textbf{spike train} $\mathsf{S}$ over $X$ is a finite set of \textbf{spike times} $t_i \in X$,
    \[
    \mathsf{S} =  \{ t_1,t_2,\dots,t_n\}.
    \]
    We use $\mathscr{S}_X$ to denote the set of all possible spike trains over $X$.
\end{definition}

\begin{remark}

Spike trains are frequently modeled mathematically as realizations of a point process over a (continuous) interval of time~\cite{kass2005statistical,truccolo2005point}. Here, we instead consider the time domain as discrete, which is appropriate for uniform time sampling, as is typical in neural recordings. Indeed, any finite collection of spike trains recorded over a continuous time window can be normalized so that spike times correspond to integers. 
\end{remark}

A fundamental difficulty with interpreting neural spike train data is that individual neurons exhibit variable responses to identical stimuli across repeated trials. This trial-to-trial variability in spike timing and rate reflects the stochastic nature of neural computation and motivates our probabilistic approach in the following definitions. From now on, we fix the discrete time domain $X = \{0,1,\dots,T\}$. 

\begin{definition}
    A \textbf{neuron} is a function $\mathbf{n}:\mathcal{I} \to \mathbb{P}(\mathscr{S}_X)$ where $\mathcal{I}=\{s_1,s_2,\dots,s_m\}$ is a finite set of \textbf{stimuli}, and $\mathbb{P}(\mathscr{S}_X)$ denotes the set of probability measures over $\mathscr{S}_X$. 
\end{definition}

\begin{remark}\label{rem:conditional_probability}
    A typical mathematical formalism for representing neurons is as a conditional probability distribution of the form $P(\mbox{response}\, |\, \mbox{stimulus})$, where the response may be a spike train or an associated statistic such as firing rate~\cite{meyer2017models,dayan2005theoretical}. Since our stimulus space $\mathcal{I}$ is assumed to be finite (an assumption which is essentially forced in any experimental setup), this conditional probability perspective is equivalent to the function definition provided above. 
\end{remark}

We are interested in the experimental setting where signals from multiple neurons are recorded for each stimulus, which is captured in the following definition.

\begin{definition}
    A \textbf{train ensemble} or \textbf{raster} of size $k \in \Z_{>0}$ is a $k$-tuple $\mathcal{R}$ of spike trains. That is, $\mathcal{R}$ is an element of the product space 
    \[
    \mathscr{S}^k_X  \coloneqq \underbrace{\mathscr{S}_X \times \cdots \times \mathscr{S}_X}_{k\text{ times}}.
    \]
    A \textbf{(neuronal) network} of size $k \in \Z_{>0}$ is a function
    \[
    \mathbf{N}:\mathcal{I} \to \mathbb{P}(\mathscr{S}_X^k),
    \]
    where $\mathcal{I}=\{s_1,s_2,\dots,s_m\}$ once again denotes a finite set of stimuli.
\end{definition}

\begin{remark}
    Neuronal networks, also referred to as \emph{neuronal populations}, are modeled as conditional distributions in the literature~\cite{paninski2007statistical,ganmor2015thesaurus}. As was similarly observed in Remark~\ref{rem:conditional_probability}, this is equivalent to our definition, as the stimulus space is assumed to be finite. 
\end{remark}

\begin{remark}\label{rem:train_ensemble_convention}
    As defined, a train ensemble $\mathcal{R} \in \mathscr{S}_X^k$ is an ordered $k$-tuple of spike trains. Below, it will be convenient to forget the ordering and consider $\mathcal{R}$ as a multiset of spike trains, as we will then endow this multiset with a pseudometric---recall that a \textbf{pseudometric} is a function satisfying the axioms of a metric, except that unequal points can receive distance zero. For the rest of the paper, we abuse notation and consider $\mathcal{R}$ as a multiset when convenient, where this interpretation will always be clear from context. Moreover, for the sake of convenience, we make the mild assumption that spike trains are not repeated in the rasters, so that $\mathcal{R}$ is always a set (rather than a multiset). 
\end{remark}

From a neuronal network $\mathbf{N}$, one obtains a $k$-tuple of neurons $(\mathbf{n}_1,\ldots,\mathbf{n}_k)$, via the associated marginal distributions for each stimulus. In practice, we have access to samples from these marginal distributions, each of which can alternatively be considered as a sample of a raster. This is illustrated schematically in Fig.~\ref{fig:Dataset Figure}.

\begin{figure}[htbp]
    \centering
    \includegraphics[width=1\linewidth]{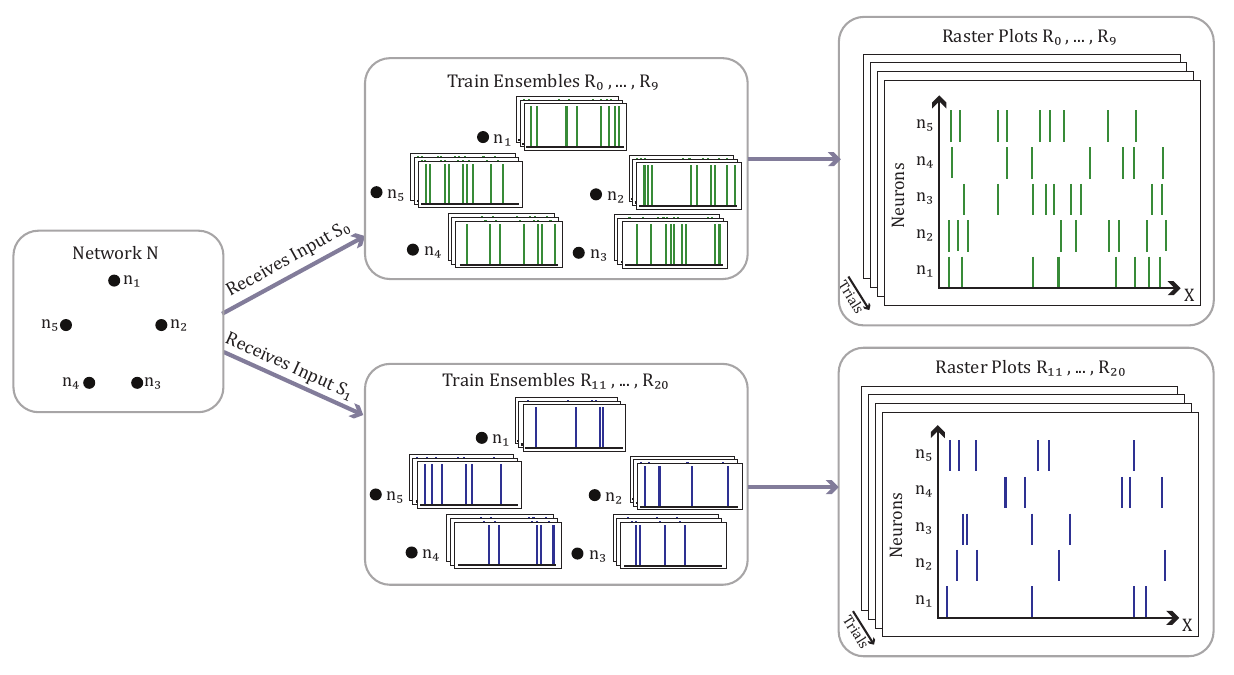}
    \caption{Example dataset under consideration. The network $\mathbf{N}$ consists of five neurons, $\mathbf{n}_1,\mathbf{n}_2,\mathbf{n}_3,\mathbf{n}_4,\mathbf{n}_5$. There are two distinct stimuli, $s_0$ and $s_1$, each presented to the network $10$ times. Each exposure results in a train ensemble (a single spike train for each neuron in the network, in this case a set of five spike trains). This is depicted in two distinct ways: (1) the resulting trains are drawn next to each neuron, and (2) the \textbf{raster plots} shown on the right, where the time domain $X$ is on the $x$-axis and the $y$-axis represents neurons with their corresponding trains. }
    \label{fig:Dataset Figure}
\end{figure}

\subsection{The Victor-Purpura Distance}
In this study, a single data point is a train ensemble $\mathcal{R} \in \mathscr{S}_X^k$ sampled from a fixed network $\mathbf{N}$ in response to a given stimulus. Our objective is to study the network's behavior from a topological perspective, which requires a notion of similarity between the spike trains of the ensemble $\mathcal{R}$. To this end, we employ the \emph{Victor-Purpura (VP) distance}~\cite{VictorPurpura1996} on the space of spike trains. The VP distance between spike trains $\mathsf{S}$ and $\mathsf{S}'$ was originally defined as an optimization problem: it is the lowest \emph{total cost} required to convert $\mathsf{S}$ to $\mathsf{S}'$ via a sequence of moves consisting of either erasing or adding spikes (each of these modifications incurs a cost of 1), or adjusting spike timings (each adjustment incurs a cost of $q$ times the distance moved, where $q \geq 0$ is a hyperparameter). We present an alternative definition of this distance using the terminology of partial bijections, borrowed from the TDA literature.  This equivalent definition of the VP distance is better suited for obtaining the theoretical results of this article.

\begin{definition} \label{definition: VP distance}
    Let $\mathsf{S} =\{t_1,t_2,\dots,t_n\}$ and $\mathsf{S}' =\{t'_1,t'_2,\dots,t'_m\}$ be two spike trains, and let $q \geq 0$ be a fixed cost parameter. A \textbf{partial bijection} from $\mathsf{S}$ to $\mathsf{S}'$, denoted $\varphi:\mathsf{S} \rightharpoonup\mathsf{S}'$, consists of a subset $\mathrm{dom}(\varphi) \subset \mathsf{S}$ called the \textbf{domain}, a subset $\mathrm{codom}(\varphi) \subset \mathsf{S}'$ called the \textbf{codomain}, and a bijective function $\varphi:\mathrm{dom}(\varphi) \to \mathrm{codom}(\varphi)$. The \textbf{$q$-cost} of a partial bijection $\varphi:\mathsf{S} \rightharpoonup\mathsf{S}'$ is given by 
    \[
\mathrm{cost}_q(\varphi) \coloneqq | \mathsf{S} \setminus \mathrm{dom}(\varphi)| + | \mathsf{S}' \setminus \mathrm{codom}(\varphi)| + q \sum_{t \in \mathrm{dom}(\varphi)}|t - \varphi(t)|,
\]
where $|\mathsf{A}|$ is used to denote the cardinality of a set $\mathsf{A}$. Finally, the \textbf{$q$-Victor-Purpura (VP) distance} between $\mathsf{S}$ and $\mathsf{S}'$ is defined by
    \begin{equation*}
        \mathrm{VP}_{q}(\mathsf{S},\mathsf{S}') \coloneqq \min_{\varphi:\mathsf{S} \rightharpoonup\mathsf{S}'} \mathrm{cost}_q(\varphi),
        \end{equation*}
where the minimum is taken over all partial bijections.    
\end{definition}

The first two terms in the cost function count the spikes that are erased in $\mathsf{S}$ and $\mathsf{S}'$, respectively, while the last term accounts for the total cost of shifting the matched spikes. Furthermore, we will 
refer to a partial bijection $\varphi:\mathsf{S} \rightharpoonup \mathsf{S}'$ as an \textbf{optimal partial bijection} if it realizes the distance.

\begin{remark} \label{remark: VP distance}
Note that although our formulation differs in presentation from that of~\cite{VictorPurpura1996}, the two can be readily shown to be equivalent.
Furthermore, it has been shown that the VP distance defines a metric on the set of spike trains $\mathscr{S}_X$~\cite{VictorPurpura1996}, i.e, 
it is positive-definite, symmetric and satisfies the triangle inequality (although, it is only a pseudometric when $q=0$, as is clear from one of the comments below). Finally, as observed in~\cite{VictorPurpura1996}, the VP distance admits specific expressions in the two extreme cases of the parameter~$q$:
\end{remark}
        \begin{enumerate}
            \item For $q = 0$, spikes may be shifted free of cost, so the VP distance depends only on the difference in the total numbers of spikes, i.e., $\mathrm{VP}_q(\mathsf{S},\mathsf{S}') = n-m$.
            \item For $q >2$, shifting a spike costs more than erasing both of the spikes. In this case, the distance formula is given by the total number of spikes (in both trains) that do not happen at the exact same time, i.e., $\mathrm{VP}_q(\mathsf{S},\mathsf{S}') =  |\mathsf{S}'\setminus \mathsf{S}| + |\mathsf{S}\setminus \mathsf{S}'|$.
        \end{enumerate}
        
\begin{figure}
    \centering
    \includegraphics[width=0.8\linewidth]{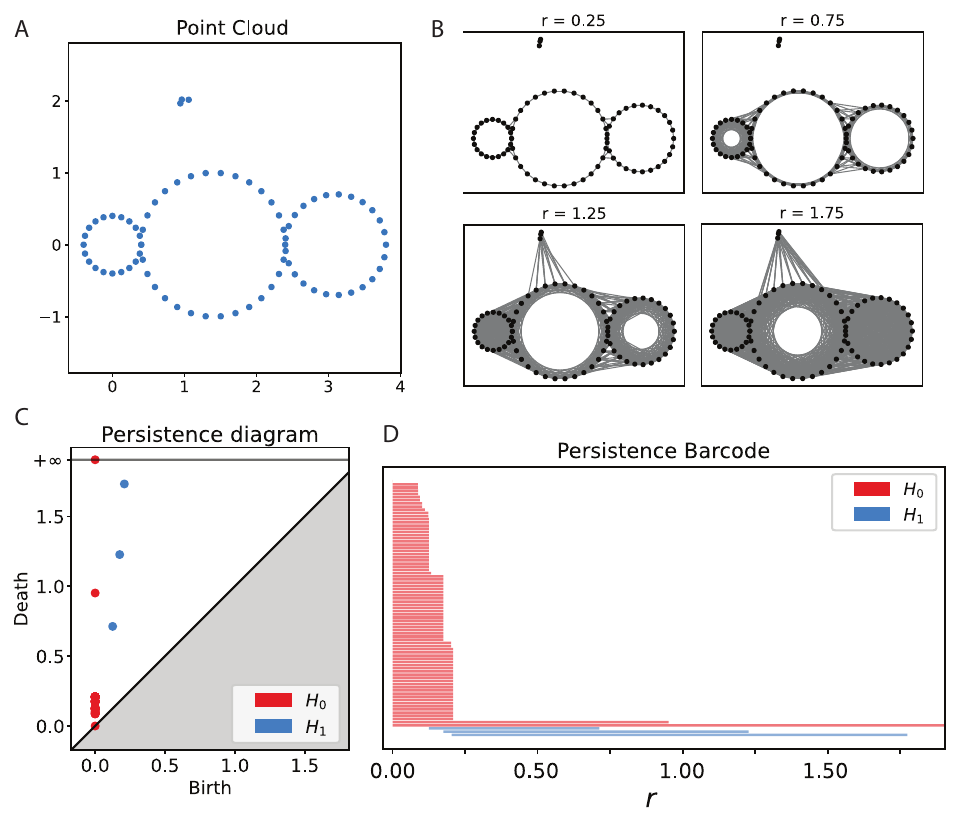}
    \caption{Illustration of persistent homology applied to a point cloud. \textbf{(A)} A point cloud in $\mathbb{R}^2$---a finite metric space $(X, d_X)$. The data appears to form three circular shapes as well as a separate cluster of three points above the main structure. \textbf{(B)} A nested sequence of Vietoris–Rips complexes $\mathbf{VR}_{r}(X)$ is shown at increasing scales $r$. As $r$ grows, more simplices are added. \textbf{(C)} The persistence diagram computed from the Vietoris–Rips filtration, where each point $(b_i, d_i)$ represents the birth and death of a topological feature. Red points correspond to $0$-dimensional features ($H_0$: connected components), and blue points represent $1$-dimensional features ($H_1$: loops). The two red points far from the diagonal suggest the presence of two connected components, while the three blue points reflect the existence of three loops in the data. \textbf{(D)} The associated persistence barcode visualizes the same features with horizontal bars indicating their lifespan. Red bars vanish as components merge, while long blue bars indicate circular features corresponding to central holes in the data that persist across a long range of parameter values, eventually disappearing near $r = 1.75$.}
    \label{fig:persistent homology example}
\end{figure}

\subsection{A Persistent Homology Pipeline for Comparing Ensembles of Spike Trains}
\label{section: pipeline}
We now regard a train ensemble $\mathcal{R} \in \mathscr{S}_X^k$ as a metric space $(\mathcal{R},\mathrm{VP}_{q})$ equipped with the VP distance described above (here, we use the convention described in Remark~\ref{rem:train_ensemble_convention}). To study the topological structure of $\mathcal{R}$, we employ \emph{persistent homology}~\cite{edelsbrunner2002topological, ZomorodianCarlsson2005, Carlsson2009}, a central tool in TDA, which we briefly review below; for detailed treatments, we refer the reader to the classic work by Carlsson~\cite{carlsson2014topological} and Edelsbrunner and Harer~\cite{EdelsbrunnerHarer2010}.

\subsubsection{Persistent Homology}
In a standard application of persistent homology, one starts with a finite metric space $(X,d_X)$ and constructs simplicial complexes at increasing scales; e.g., Vietoris-Rips complexes. This gives a nested sequence of simplicial complexes also known as a \textbf{filtration} or a \textbf{filtered simplicial complex}, which serves as an input for \textbf{persistent homology}. As the scale increases, homological features---such as clusters, loops, and voids--- appear and disappear. Persistent homology records the \textbf{birth} and \textbf{death} of these features across scales and summarizes their lifespans in a representation known as a \textbf{persistence barcode}, or equivalently a \textbf{persistence diagram}~\cite{ZomorodianCarlsson2005,carlsson2007theory}. Structurally, a persistence diagram is a (say, finite) multiset of points of the form $(b,d) \in \R \times (\R \cup \{+\infty\})$, such that $b \leq d$, where $b$ encodes the birth time of a feature and $d$ encodes the death time. This is a multiscale summary of the dataset $(X,d_X)$, where longer-lived features (i.e., with $d \gg b$) are often interpreted as signal, while short-lived ones may correspond to noise; see Fig.~\ref{fig:persistent homology example} for an illustration.  To further demonstrate the behavior of this framework we also provide an interactive exploration at \href{https://ayhncgty.github.io/Visual-TDA/docs/PH_demos/persistent_barcode_animations.html}{Persistent Homology Demonstrations}.\footnote{\url{https://ayhncgty.github.io/Visual-TDA/docs/PH_demos/persistent_barcode_animations.html}} Notably, the weighted graph example is especially relevant, as it reflects the type of data we work with—namely, weighted graphs in which nodes represent spike trains and edge weights are given by the Victor–Purpura distance.

Persistence diagrams are compared via a canonical metric known as the \textbf{bottleneck distance}~\cite{frosini2001size,Cohen-Steiner2007}, denoted $d_B$. Given two persistence diagrams $\mathcal{B}$ and $\mathcal{B}'$, this is an extended metric (i.e., it is allowed to take the value $+\infty$), defined as 
\begin{equation}\label{eqn:bottleneck_distance}
\begin{split}
d_B(\mathcal{B},\mathcal{B}') &\coloneqq \min_{\varphi:\mathcal{B} \rightharpoonup \mathcal{B}'} \max\left\{\max_{(b,d) \in \mathrm{dom}(\varphi)} \|(b,d) - \varphi((b,d))\|_\infty, \right. \\
&\qquad \qquad \qquad \qquad \left. \max_{(b,d) \in \mathcal{B} \setminus \mathrm{dom}(\varphi)} \frac{d-b}{2}, \max_{(b',d') \in \mathcal{B}' \setminus \mathrm{codom}(\varphi)} \frac{d'-b'}{2} \right\},
\end{split}
\end{equation}
where the minimum is over partial bijections (Definition~\ref{definition: VP distance}), and $\|\cdot\|_\infty$ denotes the usual $\ell_\infty$-norm, extended so that $\|(b,d)-(b',d')\|_\infty \coloneqq +\infty$ if and only if exactly one of the death times $d,d'$ is equal to $+\infty$. We observe that the structure of the bottleneck distance is similar to that of the VP distance, as formulated in Definition~\ref{definition: VP distance}.

\subsubsection{Description of the Persistent Homology Pipeline} 

\begin{figure}[htbp]
    \centering
    \includegraphics[width=1.0\linewidth]{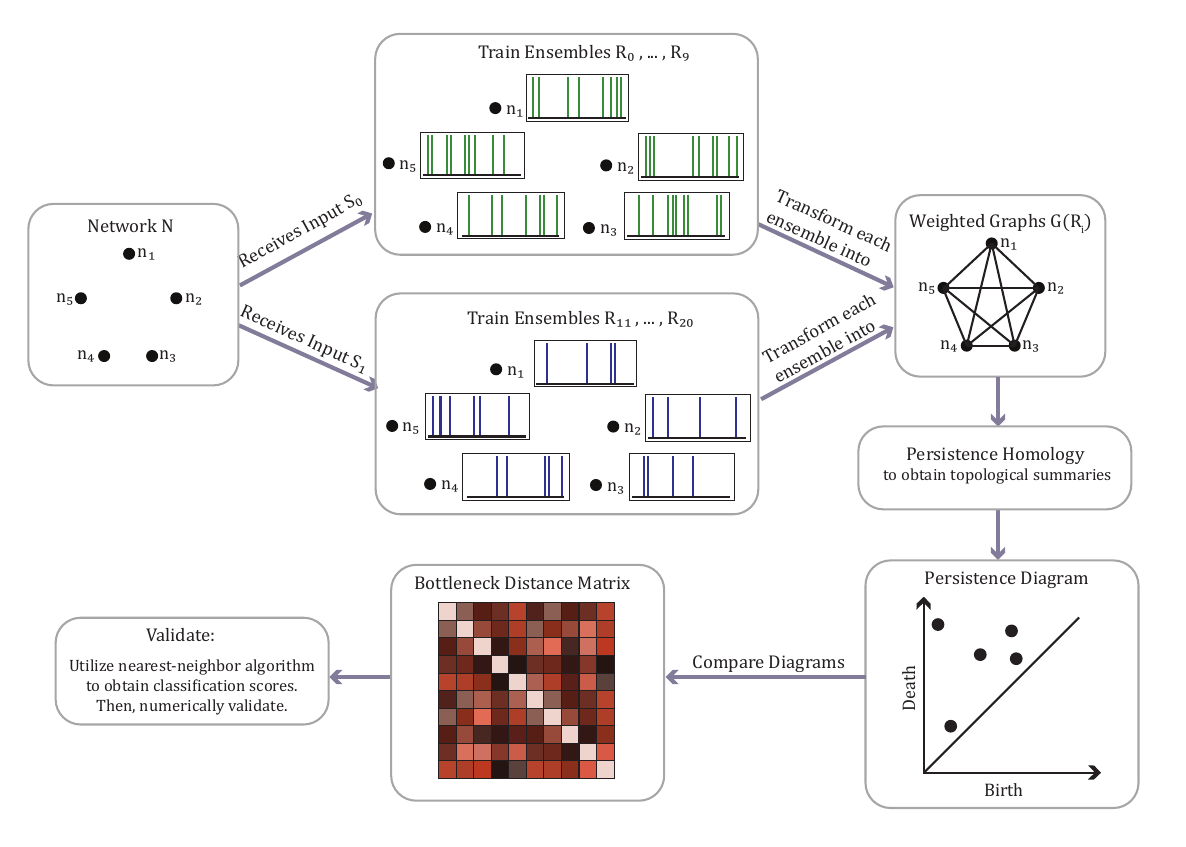}
    \caption{Pipeline schematics. The process begins with a network $\mathbf{N}$, to which different stimuli are presented multiple times. In this schematic, we consider two stimuli, $s_0$ and $s_1$, though the framework can be generalized to any finite number of stimuli. Each presentation of a stimulus yields a train ensemble $\mathcal{R}_i$, which is then converted into a weighted graph $G(\mathcal{R}_i)$. The figure shows one representative example of such graphs. All weighted graphs are transformed into a filtered simplicial complex $\mathbf{VR}_{*}$ using the Vietoris-Rips construction. By computing persistent homology, we obtain a topological summary in the form of a persistence diagram, which serves as a topological representation of the train ensembles. To analyze how much information from the network dynamics is encoded in these diagrams, we compare them using the bottleneck distance. 
    To numerically validate these findings, LOOCV is performed to obtain a classification score.}
    \label{fig:Main Pipeline}
\end{figure}

\label{subsubsection: Description of the Proposed TDA Pipeline}
TDA has been applied in several neuroscience studies, including~\cite{gardner2022toroidal,giusti2015clique,dabaghian2012topological}. In this section, we introduce a TDA-based pipeline to analyze train ensembles. Our motivation for applying TDA to this type of data stems from the intrinsic topological structure present in train ensembles that are not easily captured by traditional statistical methods.

To describe our method, we begin with a dataset of train ensembles $\mathcal{R}_i \in \mathscr{S}_X^k$ which are sampled from a fixed network $\mathbf{N}:\mathcal{I} \to \mathbb{P}(\mathscr{S}_X^k)$ under various stimuli $\mathcal{I}$. Each $\mathcal{R}_i$ is labeled according to its corresponding stimulus. Our pipeline consists of two steps:
\begin{enumerate}
    \item \textbf{Compute (fixed degree-$k$) persistent homology of each $\mathcal{R}_i$} by treating them as metric spaces $(\mathcal{R}_i, \mathrm{VP}_q)$; or as a weighted graph where nodes represent spike trains in $\mathcal{R}_i$, and edges are weighted by the VP distance. In this way, each train ensemble is represented by a persistence barcode, denoted by $\mathcal{B}_k(\mathcal{R},\mathrm{VP}_q)$, that encodes the topological structure of the neural activity. For example, for degree-0 persistent homology, the resulting barcode reflects the clustering structure of the spike trains within $\mathcal{R}_i$. Degree-1 features, on the other hand, may indicate recurrent or cyclic firing patterns—situations where the activity of one neuron triggers another in a feedback sequence that may eventually loop back to the first neuron. Higher-dimensional features may also appear, although their biological interpretation becomes increasingly speculative.
    \item \textbf{Compare persistence barcodes}: Step 1 yields a collection of persistence barcodes, each corresponding to a train ensemble. To assess the similarity between ensembles, we compare their barcodes using the bottleneck distance \eqref{eqn:bottleneck_distance}. The resulting pairwise distances are assembled into a \emph{bottleneck distance matrix}, which provides a quantitative measure of how different the underlying topological structures are. 
\end{enumerate}

We refer the reader to Fig.~\ref{fig:Main Pipeline} for an illustration of the overall pipeline. In the next section, we address several methodological questions of this approach. Subsequently, in Section~\ref{sec:experimental results and applications to real data}, we apply the method to synthetic and  biological data.

%%%%%%%%%%%%%%%%%%%%%%%%%%%%%Theoretical results%%%%%%%%%%%%%%%%%%%%%%%%%%%%%%%
\section{Methodological Results}\label{sec:results}

This section is concerned with theoretical properties of the analysis pipeline described above. We focus on \emph{stability} results, which provide guarantees that our pipeline is not extremely sensitive to hyperparameter choice or changes in the data. First, we show that the framework is stable to changes in the $q$-parameter in the Victor-Purpura distance. We then prove a general result which says that small changes to the neuronal network distributions yield small changes in the distributions of barcodes. 

\subsection{Stability Under Victor-Purpura Metric Variations}
Let $\mathcal{R}$ be a fixed train ensemble, equipped with the VP distance $\mathrm{VP}_q$ for some choice of $q \geq 0$. It is natural to ask whether the pipeline described in Section~\ref{subsubsection: Description of the Proposed TDA Pipeline} is sensitive to the choice of hyperparameter $q$. Our main theorem for this subsection shows that the resulting barcode $\mathcal{B}_k(\mathcal{R},\mathrm{VP}_q)$, varies Lipschitz continuously with respect to $q$. 

\begin{theorem}\label{thm: VP stability Theorem}
    Let $\mathcal{R} \in \mathscr{S}_X^k$ be a train ensemble. Then, for any $q,q' \geq 0$, we have
   \[
        d_B\big(\mathcal{B}_k(\mathcal{R},\mathrm{VP}_{q}),\mathcal{B}_k(\mathcal{R},\mathrm{VP}_{q'})\big) \leq \frac{T(T+1)}{2}|q-q'|.
    \]
That is, the persistent homology map $q \mapsto \mathcal{B}_k(\mathcal{R},\mathrm{VP}_{q})$ is Lipschitz.
\end{theorem}

\begin{remark}
    We make several comments regarding Theorem~\ref{thm: VP stability Theorem}.
\end{remark}
\begin{itemize}
    \item The Lipschitz constant is rather large. As we will see in the proof of Lemma~\ref{lemma: VP stability lemma}, it depends on a certain coarse estimate of the worst-case cost of a partial bijection. However, one can show that it is essentially tight---for example, one can easily construct explicit rasters which prove that the Lipschitz constant must be at least $T^2/8$. Indeed, one such construction is as follows: take $q=0$, $q'=1/2$, $k=0$, and, for simplicity, $T = 2\ell$; the bound is then achieved by the raster $\mathcal{R} = (\mathsf{S}_1,\mathsf{S}_2)$, where $\mathsf{S}_1 = \{1,2,\ldots,\ell\}$ and $\mathsf{S}_2 = \{\ell+1,\ldots,2\ell\}$. 
    \item With additional assumptions, one can improve the Lipschitz constant. For example, by imposing a sparsity assumption (i.e., an upper bound on the number of  spikes in a spike train) or assuming that $q,q' \geq 1$, one obtains a constant which scales linearly, instead of quadratically, in $T$. Rather than considering many special cases, we prefer to keep a simple statement with the  takeaway being general Lipschitz continuity.

    \item As a numerical validation, we refer the reader to Fig.~\ref{fig:vp stability}, which illustrates that small changes in $q$ values do not significantly affect the classification scores.
\end{itemize}

To prove Theorem~\ref{thm: VP stability Theorem}, we will need the following technical Lemma, which provides an upper bound on how much the VP distance can change as the parameter $q$ varies.

\begin{lemma} \label{lemma: VP stability lemma}
    Let $\mathsf{S}$ and $\mathsf{S}'$ be two spike trains over the time domain $X = \{0,1,2,\dots,T\}$. For any $q,q'\geq 0$, we have
    \[
        |\mathrm{VP}_q(\mathsf{S},\mathsf{S}')-\mathrm{VP}_{q'}(\mathsf{S},\mathsf{S}')| \leq \frac{T(T+1)}{2} |q-q'|
    \]
\end{lemma}

\begin{proof}
    We begin by noting that there are only finitely many partial bijections from $\mathsf{S}$ to $\mathsf{S}'$. Hence, there exist partial bijections $\varphi$ and $\varphi'$ that are optimal with respect to $q$ and $q'$, respectively. We first observe that
    \[
        \sum_{t\in \mathrm{dom}(\varphi)}|t - \varphi(t)| \leq T+ (T-1) + (T-2) +\cdots + 1 = \frac{T(T+1)}{2}.
    \]
    This inequality follows from the fact that for an optimal partial bijection, the maximum possible shift amount is $T$. Moreover, if $\varphi$ shifts one spike by $T$, then the next maximum possible shift is $T-1$, and so on. It follows that
    \begin{align*}
        \mathrm{cost}_{q'}(\varphi)-\mathrm{cost}_{q}(\varphi) \leq |q-q'| \sum_{t\in \mathrm{dom}(\varphi)}|t-\varphi(t)| \leq \frac{T(T+1)}{2} |q-q'|.
    \end{align*}
Using this we calculate
    \begin{align*}
        \mathrm{VP}_{q'}(\mathsf{S},\mathsf{S}') &= \mathrm{cost}_{q'}(\varphi') \\
        & \leq \mathrm{cost}_{q'}(\varphi) \\
        & \leq \mathrm{cost}_{q}(\varphi) + \frac{T(T+1)}{2} |q-q'| \\
        & =\mathrm{VP}_q(\mathsf{S},\mathsf{S}') + \frac{T(T+1)}{2} |q-q'|
    \end{align*}
    that is,
    \[  
        \mathrm{VP}_{q'}(\mathsf{S},\mathsf{S}')-\mathrm{VP}_{q}(\mathsf{S},\mathsf{S}') \leq  \frac{T(T+1)}{2} |q-q'|
    \]
    By the same argument one can also show that
    \[  
        \mathrm{VP}_{q}(\mathsf{S},\mathsf{S}')-\mathrm{VP}_{q'}(\mathsf{S},\mathsf{S}') \leq  \frac{T(T+1)}{2} |q-q'|,
    \]
    which concludes the proof.
\end{proof}

We are now ready to prove Theorem~\ref{thm: VP stability Theorem}.
\begin{proof}[Proof of Theorem~\ref{thm: VP stability Theorem}]
    A classic result in TDA, proved in~\cite{chazal2009gromov}, shows that
\begin{equation}\label{eqn:GH_stability}
d_B\big(\mathcal{B}_k(\mathcal{R},\mathrm{VP}_{q}),\mathcal{B}_k(\mathcal{R},\mathrm{VP}_{q'})\big) \leq 2 \cdot d_{GH}\big((\mathcal{R},\mathrm{VP}_{q}),(\mathcal{R},\mathrm{VP}_{q'})\big). 
    \end{equation}
    Here, the right-hand side denotes the Gromov-Hausdorff (GH) distance between the metric spaces $(\mathcal{R},\mathrm{VP}_{q})$ and $(\mathcal{R},\mathrm{VP}_{q'})$ (see, e.g.,~\cite{burago2001course} for background on Gromov-Hausdorff distance). It is immediate from the definition of GH distance that, for compact metrics $d,d'$ on a common set $X$, 
    \[
    d_{GH}\big((X,d),(X,d')\big) \leq \sup_{x,x' \in X} |d(x,x') - d'(x,x')|.
    \]
    Applying this to the VP distances, we have 
    \[
    d_{GH}\big((\mathcal{R},\mathrm{VP}_{q}),(\mathcal{R},\mathrm{VP}_{q'})\big)  \leq \max_{\mathsf{S},\mathsf{S}' \in \mathcal{R}} |\mathrm{VP}_{q}(\mathsf{S},\mathsf{S}') - \mathrm{VP}_{q'}(\mathsf{S},\mathsf{S}')|  \leq \frac{T(T+1)}{2}|q-q'|,
    \]
    where the last inequality follows from Lemma~\ref{lemma: VP stability lemma}. Putting this together with \eqref{eqn:GH_stability} completes the proof.
\end{proof}

\subsection{Stability With Respect to Changes in the Probability Measure}
In this section, we investigate the behavior of our framework with respect to changes in the probability measures governing the distribution of spike trains in a neuron or neuronal network. We start by developing some intuition for this question before we formulate our main result. Consider a network whose responses to two distinct stimuli differ only slightly (e.g., as a result of presenting nearly identical stimuli such as water at $16^{\circ}$C versus water at $16.5^{\circ}$C); our result says that train ensembles observed from such a network will yield persistence barcodes that are likewise similar.

To precisely formulate the intuition described, we first observe that our method, as described in Section~\ref{section: pipeline}, does not depend on the ordering of spike trains within an ensemble, i.e., the symmetric group $\mathsf{Sym}(k)$ acts on $\mathscr{S}_X^k$ by permuting the spike trains in an ensemble, and our framework operates on the quotient
\[
\mathscr{S}_X^k/\mathsf{Sym}(k)
\]
where two train ensembles are identified if one can be obtained from the other by reordering of its spike trains. Indeed, this is a consquence of forgetting the orderings in rasters, as described in Remark~\ref{rem:train_ensemble_convention}. This quotient structure will play a more pivotal role in this subsection, so we momentarily forgo the abuse of notation introduced in Remark~\ref{rem:train_ensemble_convention} and more carefully track the distinction between $\mathscr{S}_X^k$ and $\mathscr{S}_X^k/\mathsf{Sym}(k)$ in what follows.

For the rest of this subsection, we consider an arbitrary metric $d_{\mathscr{S}_X}$ on the set of spike trains $\mathscr{S}_X$; this could be a Victor-Purpura distance, or one of the many other spike train metrics used in the literature---see, e.g.,~\cite{victor2005spike}. The quotient  $\mathscr{S}_X^k/\mathsf{Sym}(k)$ represents unordered $k$-point subsets of $(\mathscr{S},d_{\mathscr{S}})$, and we equip it with the Hausdorff distance $d_H$:
\[
d_H (\{ \mathsf{S}_1,\mathsf{S}_2,\dots,\mathsf{S}_k\}, \{ \mathsf{S}'_1,\mathsf{S}'_2,\dots,\mathsf{S}'_k\}    ) = \max \Bigg\{
\max_{1 \leq i \leq k} \min_{1\leq j \leq k} d_{\mathscr{S}}(\mathsf{S}_i,\mathsf{S}'_j), \max_{1 \leq j \leq k} \min_{1\leq i \leq k} d_{\mathscr{S}}(\mathsf{S}'_j,\mathsf{S}_i)\Bigg\},
\]
which turns the space $\mathscr{S}_X^k/\mathsf{Sym}(k)$ into a metric space---indeed, $\mathscr{S}_X^k/\mathsf{Sym}(k)$ is a subset of the set of all subsets of $\mathscr{S}_X$, so the Hausdorff induces a metric via general principles~\cite{burago2001course} (more precisely, $\mathscr{S}_X^k/\mathsf{Sym}(k)$ can be identified with a set of \emph{multisets} of elements of $\mathscr{S}_X$, but this distinction is not important in what follows). We leverage this to define a pseudometric on $\mathscr{S}_X^k$, given by
\[
\hat{d}_H\left( 
\mathcal{R},\mathcal{R}'
\right) \coloneqq d_H \left( 
\pi (\mathcal{R}), \pi(\mathcal{R}')
\right)
\]
where $\pi:\mathscr{S}_X^k \to \mathscr{S}_X^k / \mathsf{Sym}(k)$ denotes the quotient map (that is, $\hat{d}_H$ is the pseudometric induced by considering each train ensemble as a set, rather than an ordered tuple, as in Remark~\ref{rem:train_ensemble_convention}). Note that, for a train ensemble $(\mathsf{S}_i)_{i=1}^k$, the class $\pi (\mathsf{S}_i)_{i=1}^k$ is naturally identified with the unordered set $\{ \mathsf{S}_1,\mathsf{S}_2,\dots, \mathsf{S}_k\}$. Using that $d_H$ is a metric it is easy to show that $\hat{d}_H$ is a pseudometric: it is non-negative, symmetric, and satisfies the triangle inequality. However, it is not definite---any two train ensembles that differ only by a permutation of their spike trains yield distance $0$, even though they are distinct elements of $\mathscr{S}_X^k$.

As this section is concerned with changes in probability measures, we require a metric with which to compare them. For a bounded pseudometric space $(X,d)$, we use the Wasserstein distance from optimal transport theory to compare Borel probability measures $\mu,\mu' \in \mathbb{P}(X)$. We recall the basic definition here; see, e.g.,~\cite{villani2008optimal} for a comprehensive reference. For $p \in [1,+\infty]$, the \textbf{$p$-Wasserstein distance} is given by 
\[
W_p^{d}(\mu,\mu') \coloneqq \inf_\pi \|d\|_{L^p(\pi)},
\]
where the infimum is over measure couplings and $\|\cdot\|_{L^p(\pi)}$ denotes the standard $L^p$-norm. Recall that a \textbf{measure coupling} is a Borel probability measure $\pi \in \mathbb{P}(X \times X)$ such that $(\mathrm{proj}_1)_\# \pi = \mu$ and $(\mathrm{proj}_2)_\# \pi = \mu'$, where $\mathrm{proj}_1, \mathrm{proj}_2:X \times X \to X$ are projection maps onto the left and right factors, and we use subscript $\#$ to denote the pushforward by a measurable map.

We are now ready to state our theorem:

\begin{theorem} \label{thm:stability_prob_measure}
    Let $\mu,\mu' \in \mathbb{P}(\mathscr{S}_X^k)$ be probability measures on the space of train ensembles  and let $\mathcal{B}:\mathcal{R} \mapsto \mathcal{B}_k(\mathcal{R},d_{\mathscr{S}_X})$ 
    be the map sending a train ensemble $\mathcal{R}$ to its persistence barcode (for fixed degree $k$, with respect to the Vietoris-Rips complex of 
 $d_{\mathscr{S}_X}$).
    Then
    \[
    W_p^{d_B}(\mathcal{B}_\#\mu,\mathcal{B}_\# \mu') \leq 2 \cdot W_p^{\hat{d}_H}(\mu,\mu').
    \]
\end{theorem}

Intuitively, the theorem says that if two distributions of spike train ensembles (arising as, say, stimulus responses of neurons) are similar, then so are the induced distributions of persistence diagrams. 

\begin{proof}
Let $\mathsf{Dgm}$ denote the space of persistence diagrams. We first note that the map $\mathcal{B}$ factors as 
\[\begin{tikzcd}[ampersand replacement=\&,cramped]
	{\mathscr{S}_X^k} \&\& {\mathsf{Dgm}} \\
	\\
	\& {\mathscr{S}_X^k/ \mathsf{Sym}(k)}
	\arrow["{\mathcal{B}}", from=1-1, to=1-3]
	\arrow["\pi"', from=1-1, to=3-2]
	\arrow["{\mathsf{ph}}"', from=3-2, to=1-3]
\end{tikzcd}\]
where $\mathsf{ph}$ is the Vietoris-Rips persistent homology map in the same homology degree as $\mathcal{B}$. This factorization of $\mathcal{B}$, will allow us to directly leverage the classical stability of the persistent homology map $\mathsf{ph}$ from~\cite{chazal2009gromov}.

We next show that
    \[
    \mathcal{B}:(\mathscr{S}_X^k,\hat{d}_H) \to (\mathsf{Dgm},d_B)
    \]
    is $2$-Lipschitz. Indeed, let $\mathcal{R}$ and $\mathcal{R}'$ be two train ensembles. Then
    \begin{align}
        d_B\left(\mathsf{ph} \circ \pi (\mathcal{R}),\mathsf{ph}\circ \pi(\mathcal{R}')\right) &\leq  2 \cdot d_{GH}\left(\pi (\mathcal{R}),\pi (\mathcal{R}')\right) \label{eqn:stability_prob_measure1} \\
        & \leq d_H\left(\pi (\mathcal{R}),\pi (\mathcal{R}')\right) \label{eqn:stability_prob_measure2} \\
        &= \hat{d}_H\left(\mathcal{R},\mathcal{R}'\right),  \nonumber
    \end{align}
    where the steps are justified as follows. First, recall that $d_{GH}$ denotes the Gromov-Hausdorff distance; \eqref{eqn:stability_prob_measure1} follows, as in the proof of Theorem~\ref{thm: VP stability Theorem}, by applying the main result of~\cite{chazal2009gromov}. The bound \eqref{eqn:stability_prob_measure2} follows immediately from the definition of Gromov-Hausdorff distance, and the last line is by definition of $\hat{d}_H$.  Thus, the map $\mathcal{B}$ is $2$-Lipschitz.
    
    To conclude the proof, we will use a basic, general fact about Wasserstein distances. Let $(X,d_X)$ and $(Y,d_Y)$ be bounded pseudometric spaces and let $f:X \to Y$ be a $K$-Lipschitz map (for some $K > 0$). Then, for any Borel probability measures $\mu,\mu' \in \mathbb{P}(X)$, 
    \[
    W^Y_p(f_{\#}\mu,f_{\#}\mu') \leq K \cdot W^X_p(\mu,\mu').
    \]
    This is well-known in the metric space setting (see, e.g.,~\cite[Lemma 4.11]{gomez2025metrics} for a concrete reference), and follows easily from the definition even in the pseudometric space, by following the same proof. 
    Applying this result to $\mathcal B: \mathscr{S}_X^k \to \mathsf{Dgm},d_B$ directly leads to the desired inequality
    \[
    W_p^{d_B}(\mathcal{B}_\#\mu,\mathcal{B}_\# \mu') \leq 2\cdot W_p^{\hat{d}_H}(\mu,\mu').
    \]
\end{proof}

%%%%%%%%%%%%%%%%%%%%%%%%%%%%%Experimenta results%%%%%%%%%%%%%%%%%%%%%%%%%%%%%%%
\section{Results With Synthetic and Biological Data} \label{sec:experimental results and applications to real data}

In this section, we demonstrate the capabilities of our persistent homology pipeline using both synthetic and biological data. For the synthetic data, we present two illustrative scenarios. First, 
we illustrate how the pipeline successfully classifies response stimuli even in situations where single-neuron analyses fall short. Second, we show that information may be encoded beyond simple connectivity, demonstrating a case where $0$-dimensional structure is insufficient and higher-dimensional features become essential. For the biological data, we examine a classification task involving {\em in vivo} recordings from the gustatory cortex of behaving mice. Here too, we show that the persistent homology pipeline provides superior stimulus-response discrimination compared to a single-neuron method. 
\subsection{Implementation and Experimental Setup}
Before we describe the results in more details we will briefly comment on the chosen implementation and the experimental setup for the biological data set.
\paragraph{Implementation details} All analysis was performed on laptop CPUs.
\begin{itemize}
    \item \textbf{Victor-Purpura distance.} For general VP distance computations, we used the \texttt{Elephant} package~\cite{elephant18}. In the special case $q>2$, the distance formula reduces to the simplified form
    \[
        \mathrm{VP}(\mathsf{S},\mathsf{S}')
        = |\mathsf{S}\setminus \mathsf{S}'| + |\mathsf{S}'\setminus \mathsf{S}|,
    \]
    (see remark~\ref{remark: VP distance}) for which we employed a custom implementation, denoted \texttt{VP\_trivial} in our code. 

    \item \textbf{Persistent homology.} Vietoris-Rips persistent homology was computed using the \texttt{Ripser} software package~\cite{bauer2021ripser,ctralie2018ripser}.
    \item \textbf{Bottleneck distance.} To compute the bottleneck distance for degree-$0$ diagrams in the synthetic examples, we used a custom implementation, \texttt{bottleneck\_zero}, based on~\cite[Corollary~4.15]{ayhan2025equivalence}, which is more computationally efficient than general-purpose implementations such as \texttt{persim}~\cite{scikittda}, but which applies only to degree-$0$ diagrams. Notably, in this setting, the bottleneck distance coincides with both the landscape distance~\cite{bubenik2015statistical} and the erosion distance~\cite{patel2018generalized}, which were recently proved equivalent in~\cite{ayhan2025equivalence}. For degree-$1$ barcodes, we computed the bottleneck distance using the \texttt{persim} implementation from the \texttt{scikit-tda} library~\cite{scikittda}. For experiments with recorded data from the Vincis lab, bottleneck distance for degree-$0$ diagrams was computed with the \texttt{Gudhi} package~\cite{gudhi:urm, gudhi:BottleneckDistance}.
\end{itemize}
\paragraph{Experimental Methods}
To demonstrate the pipeline with biological data we used neural recordings obtained from 16 adult wild-type C57BL/6J mice purchased from the Jackson Laboratory (Bar Harbor, ME). The primary dataset, originally published in~\cite{Bouaichi2023-ol}, consists of simultaneous spike train recordings from 433 neurons in the gustatory insular cortex (GC). During each session, mice were permitted to freely lick to receive a single droplet of deionized water delivered at one of three non-nociceptive temperatures:  cool ($14^{\circ}$C), near room temperature ($25^{\circ}$C), and warm ($36^{\circ}$C). 

Neural activity was recorded across multiple days, with each session sampling between 4–30 distinct neurons. To ensure broad sampling of neuronal populations, recording devices were gradually advanced between sessions (by approximately 200 $\mu$m for silicon probes and 100 $\mu$m for tetrode bundles), thereby yielding a unique set of neurons each day. From these sessions, we constructed 38 distinct networks, each corresponding to the simultaneous recordings from one mouse on one day. For each trial, %we extracted a 4 s temporal window of neural activity: 2 s of activity before fluid delivery and 2 s after. 
the persistent homology pipeline was applied to %a subset of this 
data from time -200 ms before the fluid delivery to 1600 ms after delivery, to focus on neural activity most relevant to stimulus processing.

Each network consisted of a group of \( k \) neurons ($k \in \{4, 5, \dots, 30 \}$) stimulated under conditions $s_i \in I$,   $I = \{s_1=14^{\circ}\text{C}, s_2=36^{\circ}\text{C}, s_3=25^{\circ}\text{C}\}$. For each network, ensembles were created by sampling spike trains across repeated trials for each stimulus. Ten randomly selected trials per condition and 20 pipeline repetitions were used, yielding for each repetition 30 train ensembles $\mathcal{R}$ (ten per stimulus condition). These ensembles served as inputs to the persistent homology pipeline.

\subsection{Results Using Synthetic Data}
In this subsection we apply the pipeline to two experiments using synthetic data.
\paragraph{TDA Discerns Patterns in Network Activity That Are Not Present in Individual Neurons}
In our first example, we construct a network of five `neurons' that was exposed to two stimuli 100 times each, creating 100 trials of synthetic data to analyze. Network responses were generated for each stimulus by constructing spike trains according to a specified firing-rate probability distribution and then perturbing them with a random noise process. 
We construct Neurons 1, 3, and 5 
to have very similar firing patterns over the entire dataset, always firing with a high rate for the first 50 ms and then spontaneously firing over the rest of the timeline. In contrast, neurons 2 and 4 are constructed to fire at the same time in response to stimulus 1, but at shifted times in response to stimulus 2. With this construction, the firing patterns of neurons 1, 3, and 5 are the same for both stimuli. For neurons 2 and 4, the two dominant spiking patterns for stimulus 1  are also the dominant spiking patterns for stimulus 2. Therefore, the spiking patterns of individual neurons do not discriminate between the two stimuli and a classification score based solely on single neuron activity should be near the value for random guessing. 
Representative raster plots of the generated spike trains are shown in Fig.~\ref{fig:synthetic_network}A. They show two example responses to each stimulus.

Next, we describe the results of our network approach, where we utilize all available trials as input to a single repetition of the pipeline.
To assess the ability of the network to discriminate between stimuli we performed a \emph{1-Nearest Neighbor} (1-NN) classification~\cite{cover1967nearest} using the bottleneck distance matrix (BDM) and evaluate its performance via \emph{Leave-One-Out Cross-Validation} (LOOCV)~\cite{stone1974cross}. The bottleneck distance matrix (BDM) is constructed such that row~$i$ of the matrix corresponds to the set of bottleneck distances between the barcode for train ensemble $\mathcal{R}_i$ and the barcodes for all other train ensembles.  In the LOOCV method, each train ensemble, $\mathcal{R}_i$,  serves once as a test set while the remaining ensembles form the training set. Each element of row $i$ of the BDM represents a comparison of a training sample with the test sample. The smallest element of the row corresponds to the spike train that is most similar to that of the test train, and so the test train is given the same stimulus label. If this label is correct, i.e., it corresponds to the stimulus that produced the test train, then the number of correctly classified samples is incremented by one.  After iterating through all samples, the final classification score is computed as the number of correctly classified samples divided by the total number of samples.

\begin{figure}[htbp]
    \centering
    \includegraphics[width=1\linewidth]{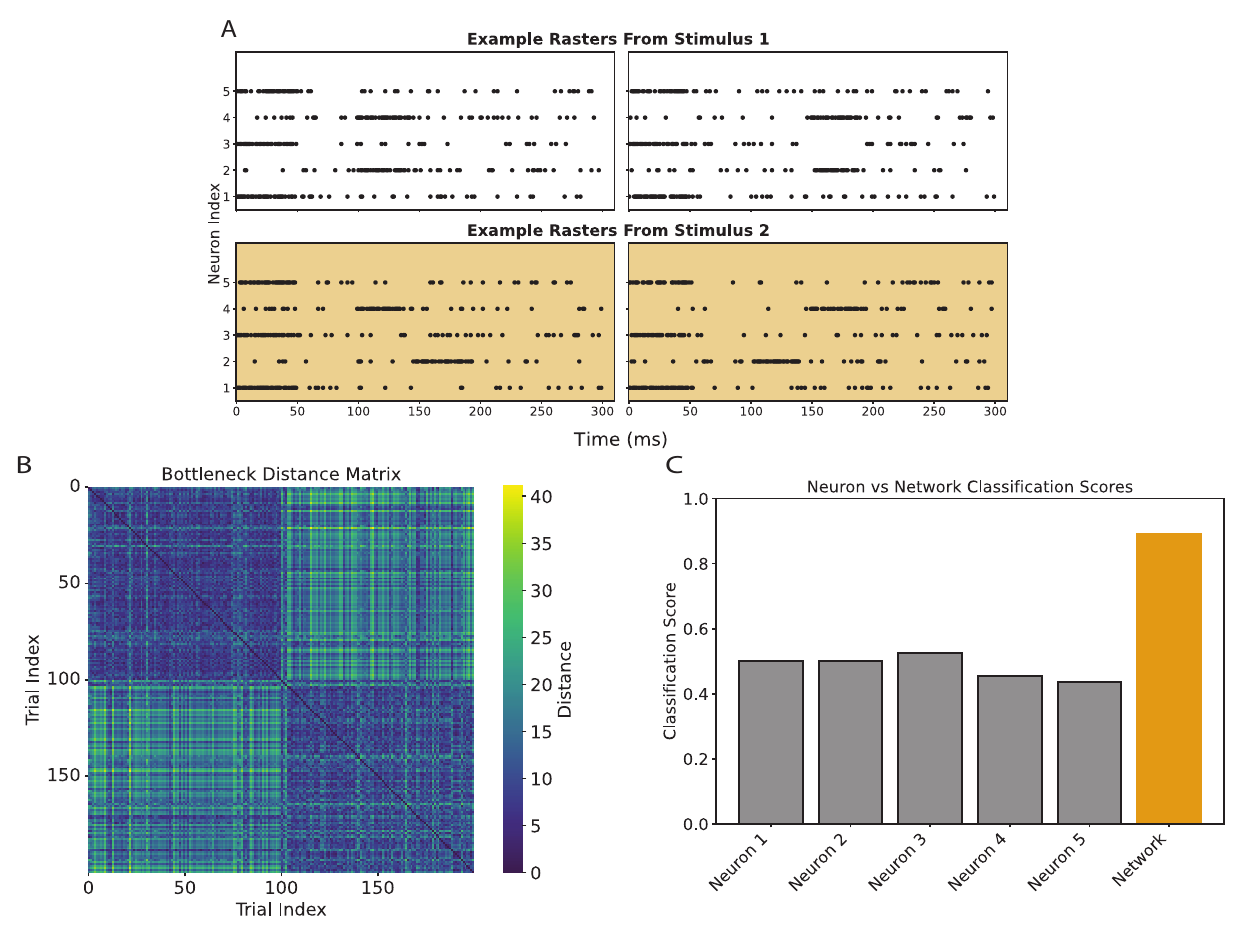}
  \caption{Synthetic network classification analysis. \textbf{(A)} A pair of raster plots generated for each of two stimuli, as described in the text. \textbf{(B)} Bottleneck distance matrix computed from the synthetic spike trains. \textbf{(C)} Classification scores for the synthetic network and its constituent neurons. Individual neuron scores (grey) have means near random guessing for all neurons. The network score (orange) has a mean significantly greater than the score for random guessing (t-test, $p<0.05$).}
    \label{fig:synthetic_network}
\end{figure}

Our experiments demonstrate that the persistent homology pipeline can discriminate response patterns that are present in the network but not in the individual neurons. That is, individual neurons fail to reliably discriminate between stimuli on their own, yet the population activity contains sufficient structure for robust stimulus classification. This indicates that the network contains more information than the sum of individual neurons and that the persistent homology pipeline can extract this information.

These findings are presented in Fig.~\ref{fig:synthetic_network}B, where 
a clear pattern in the bottleneck distance matrix is evident: it splits into four blocks, indicating there are two groups of trials that have low distance within groups and high distance between groups, on average. In Fig.~\ref{fig:synthetic_network}C, we compare the classification score of the network to the scores for individual neurons. The network achieves a classification accuracy of $0.89$ (Fig.~\ref{fig:synthetic_network}C). As expected, the individual neuron scores (computed as described in Appendix~\ref{sec:synthetic_appendix}) are close to random guessing at $0.50, 0.50, 0.53, 0.46, \text{ and }0.44$, respectively. In summary this example illustrates that the TDA analysis is effective at picking out patterns in network activity that are not present in the individual neurons. 

\paragraph{TDA Detects Higher-Order Structure Beyond Connectivity}
In our second experiment with synthetic data, we constructed a pair of train ensembles, $\mathcal{R}$ and $\mathcal{R}'$, such that their $0$-dimensional barcodes are nearly identical, i.e., $\mathcal{B}_0(\mathcal{R},\mathrm{VP}_2) \approx \mathcal{B}_0(\mathcal{R}',\mathrm{VP}_2)$, while their $1$-dimensional barcodes differ substantially. In particular, we ensured that $\mathcal{B}_1(\mathcal{R},\mathrm{VP}_2)$ contains a prominent $1$-dimensional  feature, whereas $\mathcal{B}_1(\mathcal{R}',\mathrm{VP}_2)$ is empty.

As illustrated in Figure~\ref{fig:h1_example}, the ensemble $\mathcal{R}$ exhibits a cascading activation pattern: neuron~8 activates slightly before neuron~7, which in turn precedes neuron~6, and so on, eventually looping back so that neuron~1 precedes neuron~8 again. Although this pattern would not be as apparent to the naked eye after shuffling neuron labels, persistent homology is invariant under such permutations and reliably detects the resulting loop.

In contrast, the activity in ensemble $\mathcal{R}'$ does not form a loop, leaving the $1$-dimensional barcode empty. Although their $0$-dimensional barcodes differ only by a bottleneck distance of $1$, the corresponding $1$-dimensional barcodes are separated by a distance of~$19$.

This example demonstrates that the functional networks for two train ensembles may have indistinguishable connectivity structure (as captured by $0$-dimensional persistence) yet differ in their higher-dimensional organization.  We note that higher-order features can, in principle, also be detected in more complex datasets, though their biological interpretation becomes increasingly speculative. In what follows, we restrict attention to $0$-dimensional homology, as the functional networks formed from the datasets possess sufficient informative connectivity structure to distinguish between them, but  
lack persistent $1$-dimensional features.

\begin{figure}[htbp]
    \centering
    \includegraphics[width=1\linewidth]{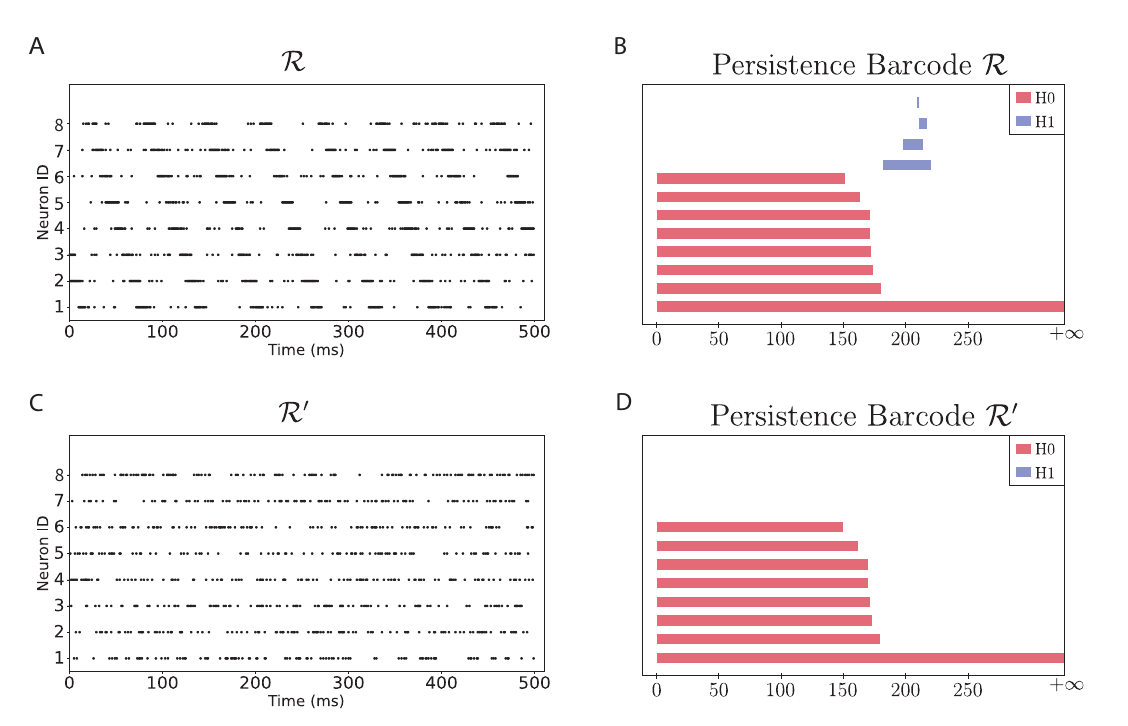}
    \caption{
    A pair of train ensembles, $\mathcal{R}$ (top) and $\mathcal{R}'$ (bottom), 
    that are similar in their $0$-dimensional structure yet differ in $1$-dimensional homology.
    \textbf{(A)} Raster plot of $\mathcal{R}$, exhibiting a repeating activation pattern across neurons.
    \textbf{(B)} Persistence barcode of $\mathcal{R}$, showing a clear $1$-dimensional feature.
    \textbf{(C)} Raster plot of $\mathcal{R}'$, lacking the recurrent pattern.
    \textbf{(D)} Persistence barcode of $\mathcal{R}'$, with no $1$-dimensional features.}
    \label{fig:h1_example}
\end{figure}

\subsection{Analysis of Neural Spike Trains From Behaving Mice}
In our next set of experiments, we investigate the capabilities of the persistent homology pipeline using experimental data recorded from behaving mice. 

\paragraph{Demonstration of VP Distance Stability} In Theorem~\ref{thm: VP stability Theorem} we showed that the classification score is stable to changes in the hyperparameter $q$ in the VP distance metric. We validate this here, using the biological data described above. In particular, we conducted a parameter sweep experiment across all gustatory cortex (GC) recordings. Each experiment corresponds to a unique mouse/day recording session, and thus to a distinct simultaneously recorded neural ensemble. In Fig.~\ref{fig:vp stability} we present the stability profile of all 38 such ensembles as the Victor-Purpura parameter $q$ is varied. For each ensemble, we applied our full persistent homology pipeline at 402 distinct $q$ values ranging from 0 to 2 in increments of $0.005$, providing a fine-grained empirical assessment of stability across parameter space. In this analysis, the score refers to the average LOOCV score computed from ten randomly selected trials per stimulus condition for each ensemble. To reduce stochastic variability we repeated the LOOCV classification for each $q$ value ten times with randomly selected trials and recorded the mean over all these trials. 

In Fig.~\ref{fig:vp stability} we display the stability landscape for all 38 experiments simultaneously over the entire parameter range $q \in [0,2]$. Here, distinct horizontal stretches of consistent coloring indicated ranges of $q$ over which the classification accuracy remains stable. While several ensembles show fluctuations for very small values of $q$, every ensemble exhibits a wide stable region for $q \geq 0.5$, providing empirical support for the robustness predicted by our theoretical results. We note, that the fluctuations for small $q$---as depicted in more detail in Panel A where we display an enlarged view of this parameter region---are not in contradiction to Theorem~\ref{thm: VP stability Theorem}. Indeed, the estimate presented in this result includes a (relatively large) Lipschitz constant $L=\frac{T(T+1)}{2}$. Thus the strong stability for larger $q$ rather suggests that the behavior in this parameter regime is better than the theoretical  guarantee. A further reason for the slight instability in the classification scores for small $q$ 
is that, at $q=0$, spike shifts can be made without incurring any penalty. Consequently, the VP distance reduces to a measure based solely on rate. When $q > 0$, however, shifting spikes does carry a cost, allowing the metric to capture temporal structure as well. Thus, this range of $q$ values represents a transition from a purely rate-based analysis to one that incorporates both rate and phase information.

\begin{figure}[htbp]
    \centering
    \includegraphics[width=1.0\linewidth]{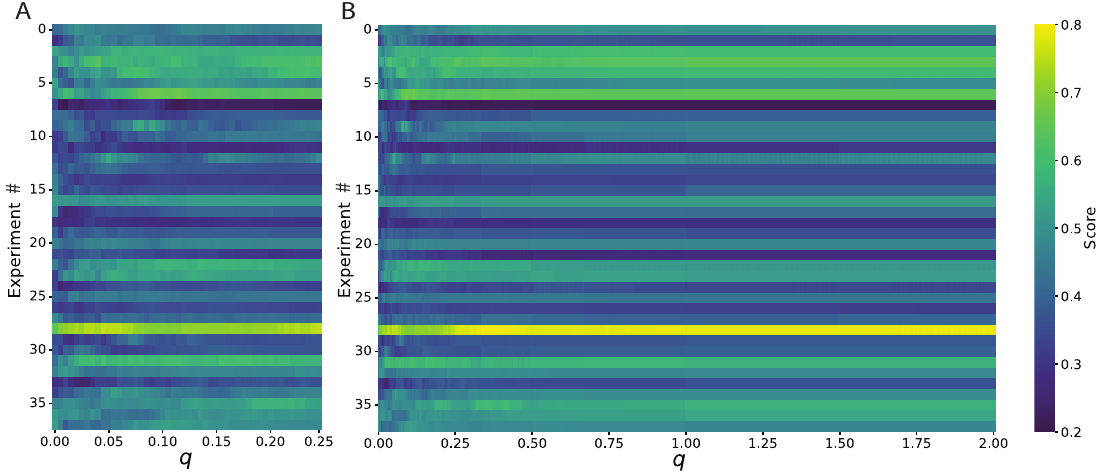}
    \caption{Stability across Experiments and hyperparameter values. Each row corresponds to a unique mouse/date combination (i.e., unique raster). Experiment labels are not the same as in figures below. \textbf{(A)} An analysis near $q=0$, where VP is primarily a measure of the spike rate, not timing. \textbf{(B)} On the wider scale of $q$ values each horizontal bar shows a consistent band color, indicating that the classification score is stable over that range of hyperparameter values. Each element corresponds to the average Leave-One-Out classification score for 401 $q$-values starting at 0 and ending at 2 in increments of 0.005. The experiments show some instability over very small $q$ values, but a large stable region as $q$ increases above 0.5. } 
    \label{fig:vp stability}
\end{figure}

\paragraph{Relationship of the Persistence Barcodes to the Biological Raster Plots}
In this section we demonstrate how the persistent homology pipeline clusters response patterns from different biological stimuli, and illustrate how persistence barcodes can be interpreted in terms of the raster plots generated from experimental recordings of neural spike timing. In Fig.~\ref{fig:mds_experiment_showcase}, we analyze an ensemble of six neurons exposed to three distinct stimuli, each presented ten times. Panel A displays the bottleneck distance matrix (BDM) obtained from the persistence diagrams of each trial, providing a pairwise measure of dissimilarity between the topological signatures of the spike train ensembles. 

To obtain a qualitative understanding of how spike train ensembles from different trials cluster together and how they separate across stimuli, we visualize the BDM using multidimensional scaling (MDS). MDS is a classical dimensionality reduction technique that embeds high-dimensional data into a lower-dimensional space while preserving pairwise distances as faithfully as possible~\cite{scikit-learn, Borg1997}. The BDM is shown in Fig.~\ref{fig:mds_experiment_showcase}A, and a 2-dimensional MDS representation based on the BDM is shown in Fig.~\ref{fig:mds_experiment_showcase}B,
with each point colored according to the stimulus. The resulting clusters demonstrate clear stimulus-specific organization.  

To illustrate how the persistence barcodes relate back to the raster plots, we consider two trials elicited by stimuli with high thermal contrast, deionized water are $14^{\circ}$C (the ``cool stimulus") and $36^{\circ}$C (``the warm stimulus"). Panels C and D show the persistence barcodes for the connected components derived from the VP distance on the spike trains. Note  that the number of bars is equal to the number of neurons in the ensemble. The $x$-axis is the distance threshold used in the filtration. At a threshold of 20 there are four bars present for the cool stimulus and three bars for the warm stimulus. This means that in the functional network graph there are four connected components for the cool stimulus and three for the warm stimulus, as shown in panels E and F, respectively. Most of the connected components are isolated nodes in the graphs, but the graph for each stimulus also contains a larger component. If nodes are connected together in the network graph it indicates that the spike patterns for those neurons are similar in spike rate and timing. 

The raster plots for the cool and warm stimuli are shown in panels G and H, respectively. In the raster plot in panel G it appears that neurons 3, 5, and 6 have similar low-frequency spiking, so the corresponding nodes form a cluster in the network graph. Neuron 4, the only high-frequency spiker, is clearly different from the rest and the corresponding node in the graph is isolated. Neurons 1 and 2 have some spike overlap, but also substantial difference in the spike timing, so the corresponding nodes in the network graph are not in the same cluster at this threshold level. Examining the raster for the warm stimulus (panel H), it is again evident that neuron 4 should not cluster with other neurons due to its much higher firing frequency, and indeed node 4 is an isolate in the network graph. Most of the other neurons have low firing frequency and some overlap in spike timing and are therefore clustered together in the network graph. The exception is neuron 1, which has low firing frequency, but the spike timing differs substantially from the other low-firing-frequency neurons. Thus, neuron 1 corresponds to an isolated node in the network graph.  

\begin{figure}[htbp]
    \centering
    \includegraphics[width=1.0\linewidth]{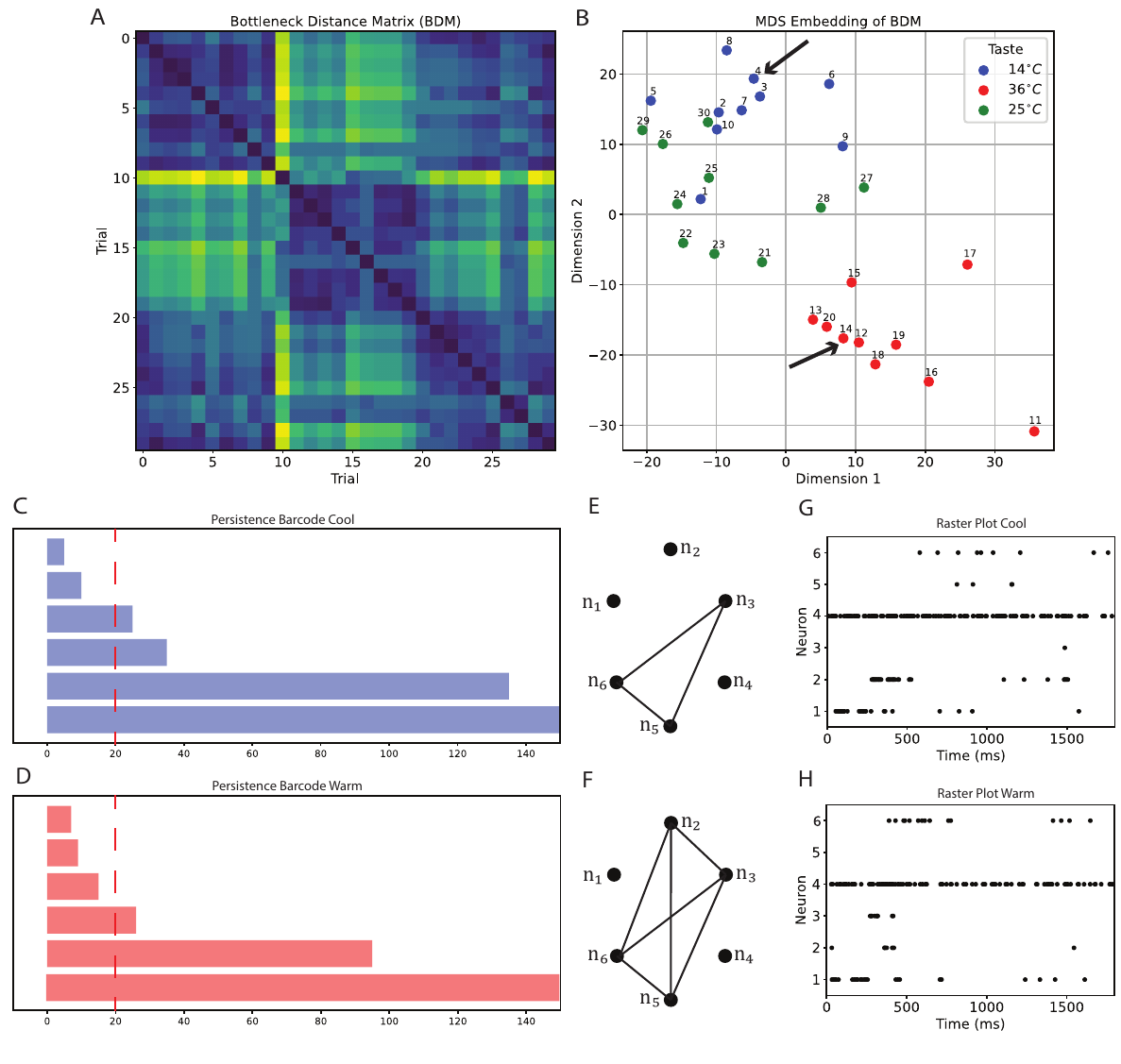}
    \caption{\textbf{(A)} Bottleneck distance matrix for an experiment in which a six neuron ensemble was exposed to three stimuli (water at $14^{\circ}$C, $36^{\circ}$C, and $25^{\circ}$C), each presented ten times. \textbf{(B)} Two-dimensional multidimensional scaling (MDS) embedding of the matrix in panel A. Arrows point to the two example trials highlighted in subsequent panels. \textbf{(C), (D)} Persistence barcodes of the connected components ($H_0$) for the two selected trials. The blue barcode corresponds to a cool stimulus, while the red corresponds to a warm stimulus. The dashed line indicates the threshold value used to form the functional network graphs. The distinct bar structures demonstrate that different stimuli evoke topologically distinguishable patterns of ensemble activity. \textbf{(E), (F)} Functional network graphs at $r=20$. The number of bars intersected by the dashed line corresponds to the number of connected components in the graph. \textbf{(G), (H)} Raster plots for the two selected trials. The spike rate and timing are the elements determining the connectivity in the functional network graphs.}
    \label{fig:mds_experiment_showcase}
\end{figure}

\paragraph{The Persistent Homology Pipeline Classification Typically Outperforms Single-Neuron Classification}

We assessed the ability of the persistent homology pipeline to distinguish response patterns to stimuli using spike train ensembles responding to deionized water at three different temperatures. Although 38 ensembles were recorded, only 33 were included in this analysis. Inclusion was restricted to ensembles in which at least one classification method exhibited performance significantly above random guessing, as determined by a one-sample t test ($p<0.05$). Each ensemble is referred to as an experiment. Network-level classification scores, utilizing 10 randomly selected experiments per stimuli, were compared to both random chance and the average performance of individual neurons in each network. For each experiment, 20 repetitions of the pipeline were employed to ensure variability across trial subsets.

Single neuron classification performance was obtained using the Bayesian rate–phase classifier described in~\cite{Nash2025-ot}, the details of which are included in Appendix~\ref{sec:gc_sna_appendix}. This analysis allowed us to assess the degree to which structured population activity, as captured by the persistent homology pipeline, enhanced stimulus discrimination beyond what was achieved from single-neuron analysis. In Fig.~\ref{fig:gctemps_session_by_session}, each experiment, representing a single session of simultaneously recorded neurons from one animal on one day, is enumerated on the x-axis, and the y-axis measures the classification score. For each experiment, the orange dot represents the mean score using the persistent homology pipeline with 20 repetitions. Each gray dot represents the single neuron classification score for each neuron in the network, and the average of these points is represented in black. 

Across the 33 experiments, 25 (~76\%) achieved higher classification accuracy with the persistent homology pipeline than with the average of all single neuron classification scores. The two methods showed a statistically significant difference in 20 of 33 experiments (two sample t-test, $p<0.05$). In these cases, the network classification score outperformed the single neuron average 14 of 20 times (70\%), whereas the mean single neuron accuracy was significantly greater than the network accuracy only six of 20 times (~30\%). Fig~\ref{fig:gctemps_session_by_session} summarizes the findings of all 33 experiments. Experiments are organized relative to a vertical divider to highlight the comparative behavior of the two approaches. To the left of the divider are the experiments in which the network classification score exceeded the mean single neuron accuracy, ordered from highest to lowest network classification score. Conversely, experiments plotted to the right of the vertical divider are those in which the average single-neuron classification outperformed the network approach. Similar to our findings in the synthetic data setting, these results demonstrate that population-level interactions captured by the persistent homology analysis can convey stimulus-related information not accessible from single-neuron responses alone.

\begin{figure}[htbp]
    \centering
    \includegraphics[width=1.0\linewidth]{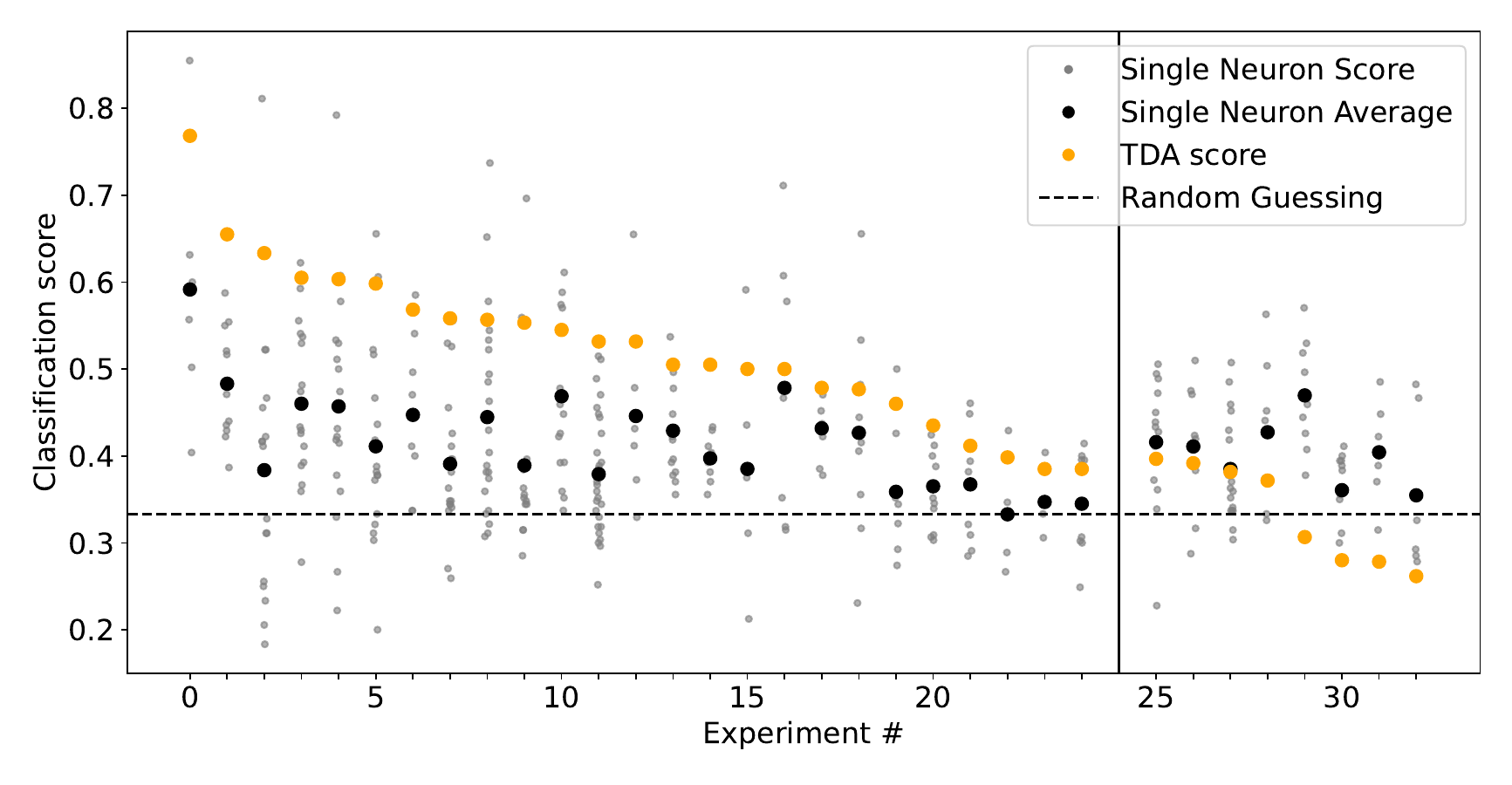}
    \caption{Classification scores using recordings from GC neurons computed using single-neuron scores and the persistent homology pipeline. Each experiment corresponds to a unique recording session. Orange points denote network classification scores captured using the persistent homology pipeline (with $q=2.005$, 10 trials, and 20 repetitions). Gray points represent scores of individual neurons in each ensemble when utilizing the Bayesian rate-phase classifier~\cite{Nash2025-ot} (with $\alpha=0.875$). Black points denote the mean of all gray scores (the average of all individual neuron scores per experiment). For all experiments left of the first vertical line, TDA classification scores outperformed the single neuron averages and were greater than random guessing (0.33). On the right of the vertical line, the single neuron average outperformed the TDA classification score.}
    \label{fig:gctemps_session_by_session}
\end{figure}

\section{Discussion}\label{sec:discussion}
While spike trains from individual neurons can be analyzed in isolation, this approach may overlook critical information encoded at the population level~\cite{buzsaki2004large,hebb1949organization,Yuste2015}. In this study, we described a new method for analyzing \emph{in-vivo} recordings from neuronal ensembles using multi-electrode arrays. Because simultaneously recorded neurons are often part of interconnected networks that process the same sensory information and contribute to coordinated behavioral responses, their collective activity patterns can reveal coding principles that single-neuron analyses cannot capture.

Our approach uses persistent homology, a topological data analysis (TDA) technique, on spike train ensembles, enabling the detection of structure in neuronal population activity that is not apparent at the level of individual neurons. We provided a formal representation of spike train ensembles as elements of a metric space endowed with the Victor-Purpura (VP) distance and established guarantees for the stability of both the distance and the resulting topological descriptors. Additionally, we demonstrated that this pipeline can successfully discriminate neural responses to sensory stimuli, thermal inputs in the gustatory cortex (GC) of behaving mice, which highlights the importance of utilizing population-level coding analysis in sensory processing.

Persistent homology captures nontrivial topological organization in neural spike train data such as clusters corresponding to distinct firing patterns for different stimuli. In both synthetic and experimental datasets, ensemble-level topological signatures enabled substantially higher stimulus classification accuracy compared to single-neuron analysis (Figs.~\ref{fig:synthetic_network},~\ref{fig:gctemps_session_by_session}). This finding suggests that there is sensory information in the structure of neural population activity that can be extracted using our persistent homology approach beyond what can be extracted from the analysis of the single-neuron spiking patterns. These results extend previous work showing that population geometry can encode stimulus features in sensory cortices~\cite{Singh2008TopologicalAnalysis, giusti2015clique, Guidolin2022-tn} and demonstrate that topological methods provide a principled away to quantify such geometry in a manner robust to noise and sampling variability. 

By treating ensembles of spike trains as elements of a metric space under the VP distance, we proved Lipschitz continuity with respect to the VP cost parameter. Additionally, we obtained a stability result concerning the probability measures that a network assigns to each stimulus, formalizing the intuition that train ensembles observed from networks responding to slightly different stimuli will yield similar topological descriptors. These results assure that small perturbations in hyperperameter choice or spike timing will not substantially alter the resulting topological summaries, which is essential for ensuring that TDA-derived features are reproducible across trials and datasets. 

Applying this framework to multi-electrode recordings from the mouse gustatory cortex revealed that response classification based on ensemble-level topological features typically outperformed single-neuron analyses in distinguishing stimuli at different temperatures. The method uncovered stimulus-specific structure within spike train ensembles, even when responses were distributed across heterogeneous neuronal populations with varying degrees of individual selectivity. By capturing these population-level dynamics, the TDA approach complements traditional statistical analyses, offering interpretable summaries of network organization and highlighting how certain ensembles exhibit coherent firing clusters in response to specific stimuli.

Advances in recording technology, such as chronically implanted high-density silicon probes \cite{jun2017fully}, allow hundreds of neurons to be monitored across multiple brain regions simultaneously in behaving animals. The persistent homology approach that we described here can be used on such data to take full advantage of the new technical developments. While we applied the pipeline to data from the gustatory cortex, it would be equally as applicable to spike train recordings from other brain regions. Temporal point process data also arise in fields outside of neuroscience, such as the timing of expression of an array of genes~\cite{Fromion13,Palande2023, Jethava2011} and the timing of transactions in financial markets~\cite{Engle03}. 

Despite these promising results, several limitations remain. The computational demands associated with persistent homology computations can limit scalability to very large datasets. Future work should explore efficient approximations to expand applicability to larger neural populations. Additionally, 
though this framework presented a result specific to the VP distance and relied on this distance for experimental results, the pipeline is general and can be utilized in conjunction with other distances (either on the space of spike trains or in other applicable settings
), such as the Van Rossum distance~\cite{Van_Rossum2001-ay}, SPIKE distance~\cite{Kreuz2007-vs}, RI-SPIKE~\cite{Satuvuori2018-nz}, a variation of Earth Mover's Distance~\cite{Sihn2019-ld}, and many others. Different distance metrics capture different aspects of neural similarity, so the suitability of the metric will depend on the questions being addressed. Although we computed some preliminary results using alternative metrics, future work that examines the impact of different metrics on the pipeline output would be a natural direction.

\appendix
\section{Further Reading}
\begin{enumerate}
    \item \cite{Centeno2022} provides a tutorial for TDA in neuroscience.
    \item \cite{curto2025topological} is a review article about TDA in neuroscience.
    \item \cite{doi:10.1073/pnas.0705546104} discusses evidence that sensory neurons function within a systems-level dynamic process best understood through distributed ensembles.
\end{enumerate}

\section{Statistical Methods For Single Neuron Analysis}

\subsection{Single Neuron Analysis Method for Fig.~\ref{fig:synthetic_network}}
\label{sec:synthetic_appendix}

For the data collected and presented in Fig.~\ref{fig:synthetic_network}C, scores for the individual neurons were calculated with a method comparable to that of the pipeline. We used the same synthetic dataset setup created for the network experiment for the individual neuron analysis. Each `neuron' responds to 2 stimuli for 100 trials, resulting in 200 trials of data for classification. We then calculated the VP distance pairwise between each trial, and used these to construct a VP distance matrix. This is analogous to the BDM in our persistent homology pipeline. 

To test the classification ability of each individual neuron, we performed leave-one-out cross-validation on the information from this VP distance matrix. Each spike train served once as a test set while the remainder of the dataset formed the training set. For the row corresponding to the test set, we found the element with the smallest value. This corresponds to the neuron with the most similar spiking pattern. The test set was then assigned the label of that most-similar element (either ``stimulus 1" or ``stimulus 2"). If the test label matches that used to create the spike train, then the classification index was incremented by 1. After iterating through all samples, so that each becomes the test element exactly once, the final classification score was computed as the number of correctly classified samples (i.e., the classification index)  divided by the total number of samples.

\subsection{Single Neuron Analysis Method for Fig.~\ref{fig:gctemps_session_by_session}}
\label{sec:gc_sna_appendix}
For the individual neuron scores in Fig.~\ref{fig:gctemps_session_by_session}, we utilized a method for understanding in-depth classification for individual coding neurons in the GC, first used in~\cite{Nash2025-ot}, where complete details of the method can be found. An abbreviated version of the method is included below. This method used information on the timing of licks, which improved the single neuron classification score~\cite{Nash2025-ot}. 

There are first some pre-processing and normalization steps, where the raw spike train recording is considered in conjunction with the lick timings. We set the initial time as the time of the first lick to the spout for which a stimulus was present. Time warping was then applied so that there was a uniform between licks (200 time units). For each trial, five lick intervals of data were examined - the first five in which stimulus was present, so that there are 1000 time units of data for analysis, which we will call a \textit{processed spike train vector}.

The trials were split into a training set (80\% of the trials) and a test set (20\% of the trials). For each stimulus, training set processed spike trains were used to construct a spike phase probability distribution. 
These were summed and divided by the number of elements in the training set to create a spike phase distributions for the training set.  A Gaussian kernel density estimator (with bandwith 5) was then used to convert the discrete distribution into a continuous probability density function for the \textit{phase} of spiking for this particular neuron/stimuli combination. To obtain a training set-based distribution for the spike rate, we used the values of the spike rate histogram from all trials in the training set and applied a gaussian kernel density estimator (with bandwidth 2), resulting in a continuous \textit{rate} probability density function. Using both of these functions, we employed a Bayesian classification scheme to obtain classification scores that quantify a neuron's ability to distinguish one stimulus from others. The test set trials were used to calculate classification scores.

Each element of the test set with $N$ spikes was interpreted as both $N$ i.i.d. samples from a phase distribution for one of the 3 stimuli ($\mu^P_i, i=1, 2, 3$) and a single sample $y$ from a rate distribution for one of the stimuli ($\mu^R_i, i=1, 2, 3$). To determine the most likely stimulus $\mathcal T_i$, we maximized the weighted sum of $P\big(\mathcal T_i|(x_1,\ldots,x_N)\big)$ and $P\big(\mathcal T_i|y\big)$ for $i=1, 2, 3$:
\begin{equation*}
   RP_i=0.875 P\big(\mathcal T_i|y\big) +(0.125) P\big(\mathcal T_i|(x_1,\ldots,x_N)\big)
   \end{equation*}
After utilizing Bayes' theorem, we obtain 
\begin{equation*}
   P\big(\mathcal T_i|(x_1,\ldots,x_N)\big)=\frac{P\big((x_1,\ldots,x_N)|\mathcal T_i\big)P\big(\mathcal T_i\big)}{\sum_{j=1}^{3}P\big((x_1,\ldots,x_N)|\mathcal T_j\big)P\big(\mathcal T_j\big)} \;\; 
\end{equation*} and
\begin{equation*}
   P\big(\mathcal T_i|y\big)=\frac{P\big(y|\mathcal T_i\big)P\big(\mathcal T_i\big)}{\sum_{j=1}^{3}P\big(y|\mathcal T_j\big)P\big(\mathcal T_j\big)}
\end{equation*}
where $P\big(\mathcal T_j\big)$ are the prior probabilities that the observation is from stimulus $\mathcal T_j$. Making use of the facts that all stimuli in the dataset are equally probable and we operate under the i.i.d assumption, we can simplify to \begin{equation*}
   P\big(\mathcal T_i|(x_1,\ldots,x_N)\big)=L^P_i/\sum_{j=1}^{3}L^P_j \;\;, \;\; P\big(\mathcal T_i|y\big)=L^R_i/\sum_{j=1}^{3}L^R_j
\end{equation*}
where $L^R_j$ and $L^P_j$ are the likelihood
\begin{equation*}
    L^R_j = \mu^R_j(y) \;\;,\;\; L^P_j = \mu^P_j(x_1)\ldots \mu^P_j(x_N) \;\;.
\end{equation*}
Using this method, we obtained a value $RP_i$ for each stimulus. The stimulus with the highest value was selected as the one most likely to have produced the test spike train sample. This process is repeated for each element of the testing set. The overall classification score, as reported in Fig.~\ref{fig:gctemps_session_by_session}, is the fraction of times this classification is correct, after this process was repeated 15 times for each neuron (each with a different random partition of the data into training and testing sets) and averaged.

\section{Supplemental Material}
The supplemental material consists of a single file, persistent\_barcode\_animations.html, an interactive browser-based viewer that visually demonstrates how filtered simplicial complexes and their associated persistence barcodes are constructed. This tool allows readers to select from several example datasets and observe, in real time, how simplices are added at increasing filtration values and how corresponding bars appear and persist in the barcode. Its inclusion is to provide an intuitive, accessible illustration of persistent homology, especially for readers who are new to the subject. 
\url{https://ayhncgty.github.io/Visual-TDA/docs/PH_demos/persistent_barcode_animations.html}

\section{Author Declarations}
\subsection{Authorship and Contributorship} C. Ayhan and A. N. Nash contributed equally. 
All authors have made substantial intellectual contributions to the study conception, methodology, formal analysis and investigation, writing, and design of the work. All authors have read and approved the final manuscript. In addition, the following contributions occurred: Funding acquisition: M. Bauer, R. Bertram, T. Needham, and R. Vincis; Supervision: M. Bauer, R. Bertram, T. Needham, and R. Vincis.

Cecilia G. Bouaichi and Katherine E. Odegaard are acknowledged for experimental data that was collected for a prior publication~\cite{bouaichi2020vincis}

ChatGPT (OpenAI GPT-5.1) assisted in formulating the optimization setup used to generate the synthetic train ensembles in Figure \ref{fig:h1_example}. All mathematical reasoning, implementation, analysis, and final decisions were made by the authors. 

\subsection{Conflicts of Interest}
The authors declare there are no conflicts of interest.

\subsection{Data \& Code Availability}
The data that supports the findings of this study is openly available in neuron-response-classification at
\newline
\url{https://github.com/vincisLab/neuron-response-classification}.
 
 The pipeline methodology, data analysis, and figure generation code that support the findings of this study are openly available in
 Persistent-Homology-Pipeline-for-Neural-Spike-Train-Data at 
\newline
\url{https://github.com/ayhncgty/Persistent-Homology-Pipeline-for-Neural-Spike-Train-Data}.

\subsection{Ethics}
This study involves the secondary analysis of animal data that were previously published and collected under ethical approval by the original investigators. No new animal experiments were conducted for this work. All data analyzed in this study were obtained from publicly available sources or published reports in which animal procedures were carried out in accordance with established institutional and national guidelines for the care and use of animals in research.

\subsection{Funding}
This research was partially supported by NSF grant DMS-2324962.

\newpage
%%%%%%%%%%%%%%%%%%%%%%%%%%%%%%%%%%%%%%%%%%%%%%%%%%%%%%%%%%%%%%%%%%%%%%%%%%%%%%%%%%%%%%%%%%%%%%%%%%%%%%%%%%%%%%%%%%%%%%%%%%%%%%%%%%%%%%
%%% Bibliography
\bibliographystyle{siam} 
\bibliography{References} 
\end{document}